\documentclass[12pt]{article}    % onecolumn (standardformat)

\usepackage[margin=0.75in]{geometry} %In Folder
\usepackage{amsmath,amssymb,amsthm}
\usepackage{color}
\usepackage{enumerate}
\newtheorem{theorem}{Theorem}[section]
\newtheorem{proposition}[theorem]{Proposition}

\newtheorem{example}[theorem]{Example}

\theoremstyle{remark}
\newtheorem{remark}{\bf Remark}

\newtheorem*{acknowledgement}{\bf Acknowledgements}

\newcommand{\be}{\begin{equation}}
\newcommand{\ee}{\end{equation}}

\newcommand{\tr}{{\rm tr}}

\newcommand{\non}{{\nonumber}}

\newcommand{\bea}{\begin{eqnarray}}
\newcommand{\eea}{\end{eqnarray}}
\newcommand{\bean}{\begin{eqnarray*}}
\newcommand{\eean}{\end{eqnarray*}}

\newcommand{\vol}{\rm vol}

\numberwithin{equation}{section}
\begin{document}
\title{\bf The largest eigenvalue distribution of the\\
Laguerre unitary ensemble}
\author{Shulin Lyu\thanks{lvshulin1989@163.com} and Yang Chen\thanks{yangbrookchen@yahoo.co.uk}\\
 Department of Mathematics, University of Macau,\\
Avenida da Universidade, Taipa, Macau, China\\}

%\authorrunning{} % if too long for running head

\date{}

\maketitle
\begin{abstract}
We study the probability that all eigenvalues of the Laguerre unitary ensemble of $n$ by $n$ matrices are in $(0,t)$, i.e., the largest eigenvalue distribution.
Associated with this probability, in the ladder operator approach for orthogonal polynomials, there are recurrence coefficients, namely $\alpha_n(t)$ and $\beta_n(t),$
as well as three auxiliary quantities, denoted by $r_n(t),~R_n(t)$ and $\sigma_n(t).$
We establish the second order differential equations for both $\beta_n(t)$ and $r_n(t).$
By investigating the soft edge scaling limit when $\alpha=O(n)$ as $n\rightarrow\infty$ or $\alpha$ is finite,
we derive a $P_{II},$ the $\sigma$-form, and the asymptotic solution of the probability.
In addition, we develop differential equations for orthogonal polynomials $P_{n}(z)$
corresponding to the largest eigenvalue distribution of LUE and GUE with $n$ finite or large.
For large $n,$ asymptotic formulas are given near the singular points of the ODE.
Moreover, we are able to deduce a particular case of Chazy's equation for $\varrho(t)=\Xi'(t)$
with $\Xi(t)$ satisfying the $\sigma$-form of $P_{IV}$ or $P_V.$
\end{abstract}

\section{Introduction}
A unitary ensemble is well defined for Hermitian matrices $M=(M_{ij})_{n\times n}$ with probability density
\bea\label{Hermitianensemble}
p(M)dM\propto e^{-\tr\;v(M)}\vol(dM),\qquad\vol(dM)=\prod\limits_{i=1}^{n}dM_{ii}\prod\limits_{1\leq j<k\leq n}d(ReM_{jk}) d (Im M_{jk}).
\eea
Here $v(M)$ is a matrix function \cite{Higham2008} defined via Jordan canonical form and $\vol(dM)$ is called the volume element \cite{Hua1963}.
The joint probability density function of the eigenvalues $\{x_j\}_{j=1}^n$ of this unitary ensemble is given in \cite{Mehta2004} by
\begin{subequations}\label{HermiteProb}
\begin{equation}
\frac{1}{D_n(a,b)}\;\frac{1}{n!}{\prod\limits_{1\leq j<k\leq n}|x_{k}-x_{j}|}^{2}\prod\limits^{n}_{j=1}w(x_j),
\end{equation}
where $D_n(a,b)$ is the normalization constant which reads
\begin{equation}\label{DnLUE}
D_n(a,b)=\frac{1}{n!}\int_{[a,b]^{n}}{\prod\limits_{1\leq j<k\leq n}|x_{k}-x_{j}|}^{2}\prod\limits^{n}_{j=1}w(x_j)dx_{j},
\end{equation}
\end{subequations}
and $w(x)=e^{-v(x)}$
is a positive weight function supported on $[a,b]$ with finite moments
$$\mu_k:=\int_a^b x^{k}w(x)dx,\qquad k=0,1,2,\cdots.$$
It is shown, in \cite{Mehta2004}, that $D_n(a,b)$ can be evaluated as the determinant of the Hankel (or moment) matrix, that is,
$$D_n(a,b)=\det\left(\mu_{i+j}\right)_{i,j=0}^{n-1}.$$
A unitary ensemble is called the Laguerre unitary ensemble (LUE) if in \eqref{Hermitianensemble}
$$v(x)=x-\alpha\ln x,$$
or, what amounts to the same thing, in \eqref{HermiteProb}
$$w(x)=x^\alpha e^{-x},\qquad x\in[0,\infty),\quad\alpha>0.$$
A special case of LUE is $M=XX^*$ and $\alpha=p-n,$
where $X=X_1+iX_2$ is an $n\times p$ $ (n\leq p)$ random matrix with each element of $X_1$ and $X_2$ chosen independently
as a Gaussian random variable, see \cite{Forrester1993, ForresterWitte2007, Goodman1963, James1964}.

Denote by $\mathbb{P}(n,t)$ the probability that the largest eigenvalue in LUE is not larger than $t,$ then
\bea\label{LargesteigenvalueProb}
\mathbb{P}(n,t)&=&\frac{D_n(t)}{D_n(0,\infty)},\nonumber
\eea
where $D_n(t):=D_n(0,t).$
Tracy and Widom \cite{Tracy1994+} have obtained the Jimbo-Miwa-Okamoto (J-M-O) $\sigma$-form \cite{Jimbo1981, Okamoto1981} of $P_V$ for
$$\sigma_n(t):=t\frac{d}{d t}\ln\mathbb{P}(n,t)$$
by making use of the Fredholm determinant. Basor and Chen \cite{BasorChen2009} have derived the same $\sigma$-form by studying the Hankel determinant $D_n(t)$
with the help of the ladder operators related to orthogonal polynomials.
In their work, another four quantities associated with $\mathbb{P}(n,t)$ are considered,
i.e. $\alpha_n(t),~\beta_n(t),~r_n(t)$ and $R_n(t),$ and the relationships between them are established.
In addition, a $P_V$ is derived for $R_n(t)$ (or $\alpha_n(t)$).
Based on these results, we obtain in this paper the second order differential equation for $\beta_n(t)$ as well as $r_n(t).$

The soft edge scaling limit of the smallest eigenvlue distribution on $(t,\infty)$ in LUE with $\alpha=\mu n=O(n)$ and
$t=\left(\sqrt{\mu +1}-1\right)^2 n-\frac{\left(\sqrt{\mu +1}-1\right)^{4/3}}{(\mu +1)^{1/6}}n^{1/3}s$
is analyzed in \cite{PerretSchehr2015}.
Concerning the largest eigenvalue distribution, we show that for $\alpha=O(n)$ or finite,
and
$$
t=c_1 n+c_2 n^{1/3}s,\qquad\sigma(s):=\frac{c_2}{c_1}\lim\limits_{n\to\infty}n^{-2/3}\sigma_{n}(t)$$
where
$$c_1=\left(\sqrt{\mu +1}+1\right)^2,\qquad c_2=\frac{\left(\sqrt{\mu +1}+1\right)^{4/3}}{(\mu +1)^{1/6}},
\qquad\mu=\begin{cases}
   \frac{\alpha}{n},       & \quad \alpha=O(n)\\
  0,  & \quad \alpha\;\text{is finite}\\
\end{cases},$$
the aforementioned $\sigma$-form of $P_V$ reduces down to the same $\sigma$-form of $P_{II}$ as presented in \cite{PerretSchehr2015}.
The $P_{V},$ the ODEs for $\beta_n(t)$ and $r_n(t)$ can likewise be reduced to a $P_{II}.$
According to the ODE for $\sigma(s),$ we are able to provide the behavior of $\mathbb{P}(n,t)$
for large $n$ when $s\rightarrow\infty$ or $s\rightarrow-\infty.$

By means of the ladder operators valid for the orthogonal polynomials $P_{n}(z)$ associated with the general weight function $w(x)=e^{-v(x)},$
we deal with our problem and the largest eigenvalue distribution of GUE, and show that
the corresponding $\phi_n(z):=e^{-v(z)/2}P_n(z)$ satisfy different second order ODEs for finite $n$ but the same one for large $n.$
In the case of large $n,$ we develop the asymptotic behavior and Taylor expansion of $\phi_n(z)$ near the singular points of the corresponding ODE.
Moreover, a Chazy's equation \cite{Cosgrove2006} is derived for $\varrho(t):=\Xi'(t)$ with $\Xi(t)$ satisfying the $\sigma$-form of $P_{IV}$ or $P_V,$
and this result is applied to different ensembles including the largest eigenvalue distribution of LUE and GUE.

This paper is built up as follows. In Section 2, we introduce the ladder operator technique and restate the results of \cite{BasorChen2009} which are used throughout this paper for further derivation.
We produce the ODE for $\beta_n(t)$ and establish a mapping for $r_n(t)$ and $R_n(t).$
The soft edge scaling limit is studied in Section 3.
The limiting behavior of $\phi_n(z)$ in the neighbourhood of the singular points is then presented in Section 4.
Finally, Section 5 is devoted to a derivation of Chazy's equations.
\section{Preliminaries}\label{finiten}
Monic polynomials $\{P_{n}(x)\}$ orthogonal with respect to a generic weight $w(x)$ on [a,b] is defined by the relations
\bea\label{orthpoly}
\int_{a}^{b}P_m(x)P_{n}(x)w(x)dx= h_n\delta_{mn},\qquad m\geq0,\quad n\geq0,
\eea
where $h_n$ is the square of the $L^2$ norm of the polynomial $P_n(x)$ and
\bea\label{Pn}
P_{n}(x) =x^{n}+ p_1(n)x^{n-1}+\cdots+P_n(0).
\eea
An immediate consequence of the orthogonality relation is the three-term recurrence relation \cite{Szego1939}
\bea\label{recurrence}
xP_{n}(x)= {P_{n+1}(x)}+\alpha_{n}P_{n}(x)+\beta_{n}P_{n-1}(x),\qquad n\geq0
\eea
with initial conditions
$$P_0(x):=1,\qquad\beta_0P_{-1}(x):=0.$$
Substituting (\ref{Pn}) into this relation gives rise to
$$
\alpha_{n}=p_1(n)-p_1(n+1),\qquad n\geq0$$
with $p_1(0):=0,$ which immediately yields
$$\sum\limits^{n-1}_{j=0}\alpha_{j}=-p_1(n).$$
From the recurrence relation \eqref{recurrence} and the orthogonality relation (\ref{orthpoly}), we get
$$\beta_{n}=\frac{h_n}{h_{n-1}}.$$
The lowering and raising ladder operators (see e.g. \cite{ChenIsmail1997}, \cite{ChenIsmail2004} for a precise statement) are
\begin{subequations}\label{ladderoperatorAnBn}
\begin{equation}\label{ladderoperator}
\begin{aligned}
 \left(\frac{d}{dz}+B_n(z)\right)P_n(z)&=\beta_n A_n(z)P_{n-1}(z),\\
 \left(\frac{d}{dz}-B_n(z)-v'(z)\right)P_{n-1}(z)&=-A_{n-1}(z)P_n(z),
\end{aligned}
\end{equation}
with
\begin{equation}\label{AnBn}
\begin{aligned}
A_n(z)&=\left .\frac{P_n^2(y)w(y)}{h_n(y-z)}\right |_{y=a}^{y=b}+\frac{1}{h_n}\int_{a}^{b}\frac{v'(z)-v'(y)}{z-y} P_{n}^2(y)w(y)dy,\\
B_n(z)&=\left .\frac{P_n(y)P_{n-1}(y)w(y)}{h_{n-1}(y-z)}\right |_{y=a}^{y=b}+\frac{1}{h_{n-1}}\int_{a}^{b}\frac{v'(z)-v'(y)}{z-y} P_{n}(y)P_{n-1}(y)w(y)dy,
\end{aligned}
\end{equation}
\end{subequations}
and $v(z):=-\ln w(z).$ The compatibility conditions ($S_1$), ($S_2$) and ($S_2'$) for the ladder operators, see \cite{ChenIts2010, ForresterWitte2007, Magnus1995}, are given by
\begin{align}
B_{n+1}(z)+B_n(z)&=(z-\alpha_n)A_n(z)-v'(z),\tag{$S_1$}\\
1+(z-\alpha_n)\left(B_{n+1}(z)-B_n(z)\right)&=\beta_{n+1}A_{n+1}(z)-\beta_n A_{n-1}(z),\tag{$S_2$}\\
B_n^2(z)+v'(z)B_n(z)+\sum\limits^{n-1}_{j=0}A_j(z)&=\beta_n
A_n(z)A_{n-1}(z).\tag{$S_2'$}
\end{align}

The discontinuous Laguerre weight
\begin{subequations}\label{discontinuousweightLaguerre}
\begin{align}
w(x)=(A+B\theta(x-t))x^\alpha e^{-x},\qquad A\geq0,\quad A+B\geq0,
\end{align}
where
\begin{equation}
\begin{aligned}
A+B\theta(x-t)=
\begin{cases}
    A+B,       & \quad \text{if } x>t\\
    A,  & \quad \text{if }  x\leq t\\
\end{cases}
\end{aligned}
\end{equation}
\end{subequations}
is investigated in \cite{BasorChen2009}, and the case where $A=0$ and $B=1$ leads to the smallest eigenvalue distribution of LUE on $(t,\infty).$
For our problem at hand, which is the case where $A=1$ and $B=-1,$ it is shown that
\begin{equation}\label{AnBnLUE}
\begin{aligned}
A_{n}(z)&=\frac{R_{n}(t)}{z-t}+\frac{1-R_{n}(t)}{z},\\
B_{n}(z)&=\frac{r_{n}(t)}{z-t}-\frac{r_{n}(t)+n}{z},
\end{aligned}
\end{equation}
with
\bea
R_{n}(t):=-\frac{P^{2}_{n}(t,t)}{h_{n}(t)}t^\alpha e^{-t},\qquad r_{n}(t):=-\frac{P_{n}(t,t)P_{n-1}(t,t)}{h_{n-1}(t)}t^\alpha e^{-t},\nonumber
\eea
where $P_j(t,t):=P_j(z,t)\mid_{z=t}.$ It should be noted that the $t$ dependence through the weight induces $t$ dependence of $P_j(x),$ $h_j$ and their allied quantities. For the sake of brevity, we shall not display the independence on $t$ for latter discussion unless we have to.

By using the compatibility conditions and taking the derivative of the orthogonality relation (\ref{orthpoly})
with respect to $t,$  in addition to a $P_V$ and the $\sigma$-form,
Basor and Chen show that $r_n$ and $R_n$ which are closely related to $\alpha_n$ and $\beta_n$
satisfy a couple of difference equations.
We recall a number of results from \cite{BasorChen2009} and make some remarks for our new observations.
\begin{proposition}\label{derivativerelation}
The relations between $\sigma_n(t):=t\frac{d}{dt}\ln\mathbb{P}(n,t)$ and other quantities are
\begin{align}
\sigma_n&=t\frac{d}{dt}\ln D_{n}=-t\sum\limits^{n-1}_{j=0}R_{j}\non\\
&=n(n+\alpha)+p_1(n)\nonumber\\
&=n(n+\alpha)+tr_n-\beta_n,\label{sigmarnbetan}\\
\sigma_n'&=r_{n},\qquad t\sigma''_n=tr_n'=\beta_{n}'.\non
\end{align}
\end{proposition}

\begin{proposition}\label{derivativerelation}
\begin{enumerate}[\rm(a)]
\item\label{Riccati}
The quantities $r_n$ and $R_n$ satisfy the following coupled Riccati equations:
\begin{align}
t r_{n}'&=\left(\frac{1}{R_{n}}+\frac{1}{R_{n}-1}\right)r_{n}^{2}+(2n+\alpha)\frac{R_{n}}{R_{n}-1}r_{n}+n(n+\alpha)\frac{R_{n}}{R_{n}-1},\label{Riccattirn}\\
t R_{n}'&=t R_{n}^2+(2n+\alpha-t)R_{n}+2r_n.\label{RiccatiRn}
\end{align}

\item The difference equations for $r_n(t)$ and $R_n(t)$ read
\bea
&&r_{n+1}+r_{n}=(t-2n-1-\alpha-t R_{n})R_{n},\nonumber\\
&&r_{n}^{2}\left(\frac{1}{R_{n}R_{n-1}}-\frac{1}{R_{n}}-\frac{1}{R_{n-1}}\right)=(2n+\alpha)r_{n}+n(n+\alpha).\nonumber
\eea

\end{enumerate}
\end{proposition}
\begin{remark}
The equations in $(b)$ can be rewritten as
\bea
x_{n+1}x_n&=&\frac{y^2_{n}-(2n+\alpha)y_{n}+n(n+\alpha)}{y_{n}^{2}},\nonumber\\
y_n+y_{n-1}&=&-\frac{(-t+2n-1+\alpha)x_{n}-(2n-1+\alpha)}{x^2_{n}-2x_{n}+1},\nonumber
\eea
where $$x_n:= 1-\frac{1}{R_{n-1}},\qquad y_n:=-r_n.$$
This mapping is very similar to (25) in \cite{GrammaticosRamani2014} which leads to discrete Painlev$\acute{e}$ equations.
\end{remark}

\begin{proposition}
\begin{enumerate}[\rm(a)]
\item The recurrence coefficients $\alpha_n$ and $\beta_n$ are expressed in terms of $r_n$ and $R_n$ as follows:
\bean
\alpha_n&=&2n+1+\alpha+t R_n,\label{alphanRn}\\
\beta_{n}&=&\frac{1}{1-R_{n}}\left((2n+\alpha)r_{n}+n(n+\alpha)+\frac{r^{2}_{n}}{R_{n}}\right).\label{betanrnRn}
\eean

\item The following equation holds
\begin{align}\label{odebetanrn}
\left((2n+\alpha)^2 -4\beta_n \right)r_n^{2}+2(2n+\alpha)(n(n+\alpha)-\beta_n)r_n+(n(n+\alpha)-\beta_n)^{2} -(\beta_n')^{2}=0.
\end{align}
\end{enumerate}
\end{proposition}
\begin{remark}
Equation (\ref{odebetanrn}) is given at the end of the proof of Thereom 6 in \cite{BasorChen2009}.
Solving for $r_n$ from it, differentiating both sides of the resulting equation with respect to t
and noting that $\beta_{n}'=t r_{n}',$
we establish the differential equation for $\beta_n=\beta_n(t):$
\bea\label{odebetan}
&&t ^2 \left(2 n^2 (\alpha +n)^2 (\alpha +2
   n)^2-8 n (\alpha +n) \left(\alpha ^2+3 n^2+3 \alpha
   n\right) \beta _n+6 (\alpha +2 n)^2 \beta
   _n^2\right.\nonumber\\
&&\left.\qquad-8 \beta _n^3+2 \left((\alpha +2 n)^2-4 \beta _n\right) \beta
   _n'{}^2+\left((\alpha +2 n)^2-4 \beta _n\right)^2 \beta _n''\right)^2\nonumber\\
&&=\left((\alpha +2 n)^4-\alpha ^2 (\alpha +2 n)t-8
   (\alpha +2 n)^2 \beta _n+16 \beta
   _n^2\right)^2 \nonumber\\
&&\qquad\cdot\left(4 \beta _n\left(n (\alpha +n)-\beta _n\right)^2+\left((\alpha +2 n)^2-4 \beta
   _n\right) \beta _n'^2\right).
\eea
\end{remark}

\begin{proposition}
\begin{enumerate}[\em (a)]

\item The quantity
$$S_{n}(t):=1-\frac{1}{R_{n}(t)}$$
satisfies the following equation
\bea\label{P5}
S''_{n}&=&\left(\frac{1}{2S_n}+\frac{1}{S_n-1}\right)\left(S'_{n}\right)^{2}-\frac{1}{t}S'_{n}-\frac{(S_{n}-1)^{2}}{t^{2}}\left(\frac{\alpha^{2}}{2}\cdot\frac{1}{S_{n}}\right)\nonumber\\
&&+(2n+1+\alpha)\frac{S_{n}}{t}-\frac{1}{2}\frac{S_{n}(S_{n}+1)}{S_{n}-1},
\eea
which is a $P_V$ \cite{Mehta2004} with
$$\alpha=0,\qquad\beta=-\frac{\alpha^{2}}{2},\qquad\gamma=2n+1+\alpha,\qquad\delta=-\frac{1}{2}.$$
\item
The differential equation for $\sigma_{n}$ reads
\begin{subequations}\label{P5largestLUE}
\begin{align}\label{JMOP5}
(t \sigma_{n}'')^{2}=(\sigma_{n}-(t-2n-\alpha) \sigma_{n}')^{2}+4\sigma_{n}'^{2}(\sigma_{n}-t \sigma_{n}'-n(n+\alpha)),
\end{align}
which is the Jimbo-Miwa-Okamoto $\sigma$-form \cite{Jimbo1981, Okamoto1981} of $P_V$ (see (\ref{JMOsigmaPV}) below) with
\begin{equation}\label{paralargest}
\begin{aligned}
\nu_1=0,\qquad\nu_{2}=n,\qquad\nu_{3}=n+\alpha.
\end{aligned}
\end{equation}
\end{subequations}
\end{enumerate}
\end{proposition}
\begin{remark}
Combining (\ref{sigmarnbetan}), (\ref{RiccatiRn}) and (\ref{betanrnRn}) gives
\begin{align}
\sigma_n&=\frac{\alpha^2}{4}\cdot\frac{R_n}{1-R_n}-\frac{1}{4}(4n+2\alpha-t) tR_n-\frac{1}{4}t^2 R_n^2-\frac{1}{4}\frac{t^2(R_n')^2}{R_n(1-R_n)}\label{sigmanRn}\\
&=-\frac{\alpha^2}{4}\cdot\frac{1}{S_n}+\frac{t }{4}\cdot\frac{4n+2\alpha-t}{S_n-1}-\frac{1}{4}\cdot\frac{t^2}{(S_n-1)^2}+\frac{1}{4}\cdot\frac{t^2 (S_n')^2}{(S_n-1)^2S_n}.\label{sigmanSn}
\end{align}
This is the desired relation for $\sigma_n$ and $S_n$ as it was demonstrated in \cite{Jimbo1981}.
\end{remark}

\section{Soft edge scaling limit: $P_{II},$ the $\sigma$-form of $P_{II}$ and the tail behavior of $\mathbb{P}(n,t)$}\label{largen}
For large $n,$ we define
$$
t=c_1 n+c_2 n^{1/3} s,
$$
where
$$
c_1=\left(\sqrt{\mu +1}+1\right)^2,\qquad
c_2=\frac{\left(\sqrt{\mu +1}+1\right)^{4/3}}{(\mu +1)^{1/6}},\qquad\mu=\begin{cases}
   \frac{\alpha}{n},       & \quad \alpha=O(n)\\
  0,  & \quad \alpha\;\text{is finite}\\
\end{cases}.
$$
In the case of finite $\alpha,$ we have
$$c_1=4,\qquad c_2=4^{2/3},\qquad t=4n+4^{2/3}n^{1/3}s.$$
We also use the following symbols:
\bean
y(s)&:=&\frac{c_2}{c_1}\lim\limits_{n\to\infty}\frac{S_{n}(t)}{n^{2/3}},\\
\sigma(s)&:=&\frac{c_2}{c_1}\lim\limits_{n\to\infty}\frac{\sigma_{n}(t)}{n^{2/3}},\\
r(s)&:=&\frac{c_2^2}{c_1}\lim\limits_{n\to\infty}\frac{r_n(t)}{n^{1/3}}.
\eean
Combining $\eqref{sigmarnbetan}$ with $\sigma_n'(t)=r_{n}(t),$ and \eqref{RiccatiRn} with \eqref{alphanRn}, we obtain
\begin{align}
r(s)&=\sigma'(s)=\frac{c_2^2}{c_1^2}\lim\limits_{n\to\infty}\frac{\beta_n(t)-n(n+\alpha)}{n^{4/3}}\label{rrnbetan}\\
&=\frac{c_1}{c_2}\lim\limits_{n\to\infty}n^{2/3}R_{n}(t)=\frac{1}{c_2}\lim\limits_{n\to\infty}\frac{\alpha_n(t)-2n-\alpha}{n^{1/3}}.\label{rRnalphan}
\end{align}
Based on the results for finite $n$ presented in the previous sections, we are able to find the relations and ODEs for the above four quantities.
The equation for $\sigma(s)$ indicates the behavior of $\mathbb{P}(n,t)$ when both $n$ and  $|s|$ large.

A $P_{II}$  and the $\sigma$-form of $P_{II}$ are established in \cite{PerretSchehr2015}
for the smallest eigenvalue distribution of LUE with $\alpha=O(n).$
The analysis there is done directly on the associated $P_V$ equation.
For our problem, we shall now present a further and more detailed investigation of $\sigma(s).$

\begin{theorem}\label{doublescalinglimit}
\begin{enumerate}[\em(i)]
\item
The variables $\sigma(s),~y(s),$ and $r(s)$ are connected by the relations
\begin{align}
\sigma(s)&=-\frac{s}{y(s)}-\frac{1}{y^2(s)}+\frac{ \left(y'(s)\right)^2}{4y^3(s)}\label{sigmay}\\
&=-\frac{r'(s)^2}{4 r(s)}-r(s)^2+s r(s),\label{sigmar}\\
r(s)&=-\frac{1}{y(s)}.\label{sigma'y}
\end{align}
\item The following equation is valid
\bea\label{P5inf}
y''(s)=\frac{3}{2}\cdot\frac{y'(s)^2}{y(s)}-2s y(s)-4,
\eea
in which we introduce
$$y(s)=w^{-2}(s),$$
so that $w(s)$ satisfies the $P_{II}$ with $\alpha=0,$ namely
$$
w''(s)=2w^{3}(s)+s w(s).
$$
\item The equation for $\sigma(s)$ reads
\bea\label{JMOP5inf}
\sigma ''(s)^2+4\sigma '(s)\left(\sigma '(s)^2-s\sigma '(s)+\sigma (s)\right)=0,
\eea
which can be brought into the $\sigma$-form of $P_{II}$ with $\theta=0$ by making the replacements
$s\rightarrow-2^{-1/3}s$ and $\sigma (s)\rightarrow-2^{1/3}\sigma (s).$ Moreover, for $r(s),$ we have
\bea\label{limitrn}
r''(s)=\frac{1}{2}\cdot\frac{r'(s)^2}{r(s)}+2sr(s)-4r^2(s).
\eea
\end{enumerate}
\end{theorem}
\begin{proof}
(\ref{sigmay}) can be verified by (\ref{sigmanSn}).
To be specific, substituting $\sigma_n(t)$ by $\frac{c_1}{c_2}n^{2/3}\sigma(s),$
and $S_n(t)$ by $\frac{c_1}{c_2}n^{2/3}y(s)$ in (\ref{sigmanSn}),
we obtain (\ref{sigmay}) as the only retained leading order term when $n\rightarrow\infty.$
\eqref{sigma'y} is a direct consequence of the first expression in \eqref{rRnalphan}, which
combined with \eqref{sigmay} results in \eqref{sigmar}.
Finally, equation (\ref{P5inf}) arises from (\ref{P5}), (\ref{JMOP5inf}) from \eqref{JMOP5},
and \eqref{limitrn} from  \eqref{odebetan} (as well as from \eqref{odern}, see below).
In case \eqref{limitrn} we take into account \eqref{rrnbetan}.
\end{proof}

\begin{remark} We can prove \eqref{P5inf}-\eqref{limitrn} in another way, by use of \eqref{sigmay}-\eqref{sigma'y} and $\sigma'(s)=r(s).$
Indeed, differentiating both sides of \eqref{sigmay}, in view of \eqref{sigma'y}, we get \eqref{P5inf}.
Replacing $r(s)$ by $\sigma'(s)$ and $r'(s)$ by $\sigma''(s)$ in \eqref{sigmar} yields \eqref{JMOP5inf}.
Solve for $\sigma(s)$ from \eqref{JMOP5inf}, then differentiation gives \eqref{limitrn}.
\end{remark}

Now we go ahead with the evaluation of $\mathbb{P}(n,t).$ The definitions of $\sigma_n(t)$ and $\sigma(s)$ imply
\begin{align}
\sigma(s)=\frac{d}{ds}\ln\mathbb{\hat{P}}(s),\nonumber
\end{align}
where $\mathbb{\hat{P}}(s):=\lim\limits_{n\to\infty}\mathbb{P}(n,c_1 n+c_2 n^{1/3} s).$
To obtain the expansion of $\sigma(s)$ as $s\rightarrow-\infty,$ we assume
$$\sigma(s)=\lambda_2s^2+\lambda_1s+\sum\limits^{\infty}_{k=0}\frac{d_k}{s^{k}}.$$
Substituting this into (\ref{JMOP5inf}) yields
\bean
\sigma (s)=\frac{s^2}{4}-\frac{1}{8 s}+\frac{9}{64
   s^4}-\frac{189}{128 s^7}+\frac{21663}{512 s^{10}}+O\left(\frac{1}{s^{13}}\right),
\eean
from which follows
\bean
\mathbb{\hat{P}}(s)&=&\iota_1\frac{\exp\left(\frac{s^3}{12}\right)}{(-s)^{1/8}}\exp\left(-\frac{3}{64 s^3}+\frac{63}{256
   s^6}-\frac{2407}{512 s^9}+O\left(\frac{1}{s^{12}}\right)\right)\\
&=&\iota_1\frac{\exp\left(\frac{s^3}{12}\right)}{(-s)^{1/8}}\left(1-\frac{3}{2^{6} s^3}+\frac{2025}{2^{13} s^6}-\frac{2470825}{2^{19}
   s^9}+\frac{26389914075}{2^{27} s^{12}}+O\left(\frac{1}{s^{12}}\right)\right),
\eean
where the second equality results from the Taylor series for the exponential function.
Here $\iota_1$ is the normalization constant and is given in \cite{Tracy2009} by
$$\iota_1=2^{1/24}e^{\varsigma'(-1)},$$
where $\varsigma'(-1)$ is the derivative of the Riemann zeta function evaluated at -1. The above result agrees with (1.19) in \cite{Tracy1994}, since $w(z)$ plays the same role as $q(s;\lambda)$ there.
In fact, because of (\ref{sigma'y}) and $y(s)=w^{-2}(s),$ we have
$$
\sigma'(s)=-w^{2}(s).
$$
As $s\rightarrow\infty,$ to continue, we write
$$\mathbb{\hat{P}}(s)=1-\varepsilon f(s),$$
where $\varepsilon>0$ is sufficiently small and $f(s)>0,$ then
$$\sigma(s)=\frac{\varepsilon f'(s)}{1-\varepsilon f(s)}.$$
Hence from the left hand side of \eqref{JMOP5inf} there follows a quadratic polynomial in $\varepsilon$
and the constant term of which has to be 0, namely
$$
f^{(3)}(s)^2-4 s f''(s)^2+4 f'(s) f''(s)=0.
$$
If we define
$$h(s):=\frac{f'(s)}{f(s)}=\frac{d}{ds}\ln f(s),$$
then
\bean
\left(h''(s)+3 h(s)
   h'(s)+h^3(s)\right)^2-4 s \left(h'(s)+h^2(s)\right)^2+4 h(s)
   \left(h'(s)+h^2(s)\right)=0.
\eean
Suppose
$$h(s)=\lambda_0 s^{1/2}+\sum\limits^{\infty}_{k=0}\frac{\nu_k}{s^{k/2}},$$
then the coefficients can be determined by use of the preceding equation for $h(s)$ and we obtain
$$h(s)=-2s^{1/2}-\frac{3}{2s}+\frac{35}{16s^{5/2}}+O\left(\frac{1}{s^4}\right),$$
which implies
$$\mathbb{\hat{P}}(s)=1-\iota_2\frac{\exp\left(-\frac{4}{3}s^{3/2}\right)}{s^{3/2}}\left(1-\frac{35}{24s^{3/2}}
  +O\left(\frac{1}{s^{3}}\right)\right).$$
Referring to \cite{Tracy2009}, we see that $\iota_2=\frac{1}{16\pi}.$
\section{The behavior of $z^{\alpha/2}e^{-z/2}P_n(z)$ on $(0,t)$ and of $e^{-z^2/2}P_n(z)$ on $(-\infty,t)$ for finite $n$ and large $n$}
We now turn our attention to the ladder operators given by \eqref{ladderoperatorAnBn}
to develop the differential equation for orthogonal polynomials $P_n(z)$ defined by (\ref{orthpoly}) and (\ref{Pn}).
Eliminating $P_{n-1}(z)$ from the lowering and raising operators, incorporated with $(S_2'),$ produces
\bea\label{odePn}
P_n''-\left(\frac{A_n'}{A_n}+
   v'\right)P_n'+\left(B_n'- \frac{A_n' }{A_n}B_n+\sum\limits^{n-1}_{j=0}A_j\right)P_n=0,
\eea
which is presented in \cite{BasorChen2009} and \cite{BasorChenEhrhardt2010}.
To continue, we set, with suitable continuation in $z,$
$$P_n(z)=e^{v(z)/2}\phi_n(z),\qquad v(z)=-\ln w(z),$$
and, in light of (\ref{odePn}), establish
\bea\label{odephin}
\phi_{n}''-\frac{A_n'}{A_n}\phi_{n}'+\left(B_n'-\frac{A_n'}{A_n}\left(B_n+\frac{v'}{2}\right)+\sum\limits^{n-1}_{j=0}A_j+\frac{v''}{2}-\frac{v'^2}{4}\right)\phi_{n}=0.
\eea

Below we will apply the above equation to the largest eigenvalue distribution of LUE on $(0,t)$ and of GUE on $(-\infty,t).$
For the Laguerre case, that is,
$$\phi_n(z)=z^{\alpha/2}e^{-z/2}P_n(z),$$
using the same notations as in Section \ref{finiten} for finite $n$ and as in Section \ref{largen} for large $n$,
we have the following results :
\begin{theorem}
\begin{enumerate}[\em (i)]
\item For finite $n$
\begin{subequations}\label{LUEfinite}
\begin{equation}\label{finitephin}
\begin{aligned}
\phi_{n}''(z)&+\left(\frac{1}{z}+\frac{1}{z-t}-\frac{1}{z-t+t R_n(t)}\right)\phi_{n}'(z)\\
&+\left(\frac{-\frac{1}{4}\alpha^2}{z^2}+\frac{\frac{1}{2}(2n+\alpha+1)-\kappa_1(t)-\kappa_2(t)}{z}+\frac{\kappa_1(t)}{z-t}+\frac{\kappa_2(t)}{z-t+t R_n(t)}-\frac{1}{4}\right)\phi_{n}(z)=0,
\end{aligned}
\end{equation}
%\begin{equation}\label{finitephin}
%\begin{aligned}
%\phi_{n}''(z)&+\left(\frac{1}{z}+\frac{1}{z-t}-\frac{1}{z-t+t R_n(t)}\right)\phi_{n}'(z)\\
%&+\left(\frac{-\frac{1}{4}\alpha^2}{z^2}+\frac{\frac{1}{2}(2n+\alpha+1)-(\frac{1}{2}R_n-\frac{\sigma
%   _n}{t}+\frac{R_n'}{2(1-R_n)})}{z}+\frac{\frac{1}{2}R_n-\frac{R_n'}{2R_n}-\frac{\sigma
%   _n}{t}}{z-t}+\frac{\frac{R_n'}{2 R_n}+\frac{R_n'}{2(1-R_n)}}{z-t+t R_n}-\frac{1}{4}\right)\phi_{n}(z)=0,
%\end{aligned}
%\end{equation}
where
\begin{equation}\label{finitepara}
\begin{aligned}
\kappa_1(t)=\frac{1}{2}R_n(t)-\frac{R_n'(t)}{2R_n(t)}-\frac{\sigma
   _n(t)}{t},\qquad
\kappa_2(t)=\frac{R_n'(t)}{2 R_n(t)}+\frac{R_n'(t)}{2(1-R_n(t))},
\end{aligned}
\end{equation}
\end{subequations}
and $\sigma_n(t)$ can be expressed in terms of $R_n(t)$ by using \eqref{sigmanRn}.
\item For $z=c_1 n+c_2 n^{1/3}z^*,$  and $\phi(z^*):=\lim\limits_{n\to\infty}\phi_n(z)$
\begin{equation}\label{inftyphi}
\begin{aligned}
\phi''(z^*)&+\left(\frac{1}{z^*-s}-\frac{1}{z^*-s+r(s)}\right)\phi'(z^*)\\
&+\left(\frac{\frac{r'(s)^2}{4 r(s)}-\frac{r'(s)}{2 r(s)}-s r(s)+r^2(s)}{z^*-s}+\frac{\frac{r'(s)}{2 r(s)}}{z^*-s+r(s)}-z^*\right)\phi(z^*)=0.
\end{aligned}
\end{equation}
\end{enumerate}
\end{theorem}
\begin{proof}
\begin{itemize}
\item[(i)]
Substituting \eqref{AnBnLUE} and $v(z)=z-\alpha\ln z$ in (\ref{odephin}) furnishes
\begin{equation}
\begin{aligned}
\phi_{n}''(z)&+\left(\frac{1}{z}+\frac{1}{z-t}-\frac{1}{z-t+t R_n(t)}\right)\phi_{n}'(z)\non\\
&+\left(\frac{-\frac{1}{4}\alpha^2}{z^2}+\frac{\frac{1}{2}(2n+\alpha+1)-\kappa_1(t)-\kappa_2(t)}{z}+\frac{\kappa_1(t)}{z-t}+\frac{\kappa_2(t)}{z-t+t R_n(t)}-\frac{1}{4}\right)\phi_{n}(z)=0,
\end{aligned}
\end{equation}
where
\bean
\kappa_1(t)&=&-\frac{1}{t}\left(\frac{r_n(t)}{R_n(t)}+\frac{\alpha}{2}+n\right)+\sum\limits^{n-1}_{j=0}R_{j}(t)+\frac{1}{2},\\
\kappa_2(t)&=&\frac{1}{t-t R_n(t)}\left(\frac{r_n(t)}{R_n(t)}+\frac{\alpha}{2}+n\right)-\frac{1}{2}.
\eean
Applying \eqref{RiccatiRn} to cancel $r_n(t)$ in $\kappa_1(t)$ and $\kappa_2(t),$
and taking into consideration $\sum\limits^{n-1}_{j=0}R_{j}(t)=-\frac{\sigma_n(t)}{t},$
we get \eqref{finitepara} and hence \eqref{LUEfinite}.
\item[(ii)]
Using the first equality of \eqref{rRnalphan} and the definition of $\sigma(s),$ we can show that
$$c_2\lim\limits_{n\to\infty}n^{1/3}\kappa_1(t)=-\frac{r'(s)}{2r(s)}-\sigma(s),\qquad c_2\lim\limits_{n\to\infty}n^{1/3}\kappa_2(t)=\frac{r'(s)}{2r(s)}.$$
Denote the coefficient of $\phi_n(z)$ in \eqref{finitephin} by $C_L,$ that is,
$$C_L=\frac{-\frac{1}{4}\alpha^2}{z^2}+\frac{\frac{1}{2}(2n+\alpha+1)}{z}-\frac{1}{4}-\frac{\kappa_1(t)+\kappa_2(t)}{z}+\frac{\kappa_1(t)}{z-t}+\frac{\kappa_2(t)}{z-t+t R_n(t)}.$$
Then, for $z=c_1 n+c_2 n^{1/3}z^*$ and $t=c_1 n+c_2 n^{1/3}s,$ we have
\bean
c_2^2\lim\limits_{n\to\infty}n^{2/3}C_L&=&c_2^2\lim\limits_{n\to\infty}n^{2/3}\left(\frac{-\frac{1}{4}\alpha^2}{z^2}+\frac{\frac{1}{2}(2n+\alpha+1)}{z}-\frac{1}{4}\right)\\
&&-c_2^2\lim\limits_{n\to\infty}n^{2/3}\frac{\kappa_1(t)+\kappa_2(t)}{z}+c_2^2\lim\limits_{n\to\infty}n^{2/3}\frac{\kappa_1(t)}{z-t}\\
&&+c_2^2\lim\limits_{n\to\infty}n^{2/3}\frac{\kappa_2(t)}{z-t+t R_n(t)}\\
&=&-z^*-0+\frac{-\frac{r'(s)}{2 r(s)}-\sigma(s)}{z^*-s}+\frac{\frac{r'(s)}{2 r(s)}}{z^*-s+r(s)}\\
&=&\frac{-\frac{r'(s)}{2 r(s)}-\sigma(s)}{z^*-s}+\frac{\frac{r'(s)}{2 r(s)}}{z^*-s+r(s)}-z^*.
\eean
Noting that, for $\phi(z^*):=\lim\limits_{n\to\infty}\phi_n(z),$
$$\phi'(z^*)=c_2\lim\limits_{n\to\infty} n^{1/3}\phi_{n}'(z),
\qquad\phi''(z^*)=c_2^2\lim\limits_{n\to\infty} n^{2/3}\phi_{n}''(z),$$
according to \eqref{finitephin} multiplied by $c_2^2n^{2/3},$ as $n\rightarrow\infty,$
we conclude that \eqref{inftyphi} is valid.
\end{itemize}
\end{proof}

Now we consider the largest eigenvalue distribution of GUE on $(-\infty,t).$
Let $P_n(x)$ be monic polynomials orthogonal with respect to the deformed Hermite weight with one jump
$$w(x)=2e^{-x^2}\left(1-\theta(x-t)\right)=\begin{cases}
   0,       & \quad \text{if } x>t\\
  2e^{-x^2},  & \quad \text{if }  x\leq t\\
\end{cases},$$
namely
$$
\int_{-\infty}^{\infty}P_m(x)P_{n}(x)w(x)dx= 2h_n(t)\delta_{mn},
$$
or, what amounts to the same thing,
$$
\int_{-\infty}^{t}P_m(x)P_{n}(x)e^{-x^2}dx= h_n(t)\delta_{mn}.$$
Consequently, the results in \cite{ChenPruessner2005} with $\beta=-2$ there in the weight function are valid
for our $P_n(x)$ which corresponds to the largest eigenvalue distribution of GUE on $(-\infty,t).$ To begin with,
\bea\label{GUEAnBn}
A_n(z)=2\left(\frac{\alpha_n(t)}{z-t}+1\right),\qquad B_n(z)=\frac{r_n(t)}{z-t},
\eea
where $\alpha_n(t)$ is the recurrence coefficient, that is,
$$
zP_n(z)=P_{n+1}(z)+\alpha_n(t)P_n(z)+\beta_n(t)P_{n-1}(z),
$$
and $r_n(t)$ is defined by
$$
r_{n}(t):=-\frac{P_{n}(t,t)P_{n-1}(t,t)}{h_{n-1}(t)}e^{-t^2}.
$$
Based on the results in \cite{ChenPruessner2005}, for
$$\phi_n(z)=e^{-z^2/2}P_n(z),$$
we obtain the following analogue of the preceding theorem:
\begin{theorem}
\begin{enumerate}[\em (i)]
\item For finite $n$
\begin{subequations}\label{odephialphan}
\begin{equation}\label{GUEphin}
\begin{aligned}
\phi_{n}''(z)&+\left(\frac{1}{z-t}-\frac{1}{z-t+\alpha_n(t)}\right)\phi_{n}'(z)\\
&+\left(\frac{\kappa_3(t)}{z-t}+\frac{\kappa_4(t)}{z-t+\alpha_n(t)}+2n+1-z^2\right)\phi_{n}(z)=0,
\end{aligned}
\end{equation}
where
\begin{equation}\label{finiteGUEpara}
\begin{aligned}
\kappa_3(t)&=\frac{\left(\alpha_n'(t)\right)^2}{4\alpha_n(t)}-\frac{\alpha_n'(t)}{2\alpha_n(t)}-\alpha_n^3(t)+2t\alpha_n^2(t)+(-t^2+2n+1)\alpha_n(t),\qquad \kappa_4(t)=\frac{\alpha_n'(t)}{2\alpha_n(t)},
\end{aligned}
\end{equation}
\end{subequations}
and $\alpha_n(t)$ satisfies
\bea\label{GUEalphan}
\alpha_n''(t)=\frac{\left(\alpha_n'(t)\right)^2}{2\alpha_n(t)}+6\alpha_n^3(t)-8t\alpha_n^2(t)+2(t^2-2n-1)\alpha_n(t).
\eea
\item\label{limitphinGUE}
For $z=\sqrt{2n}+\frac{z^*}{\sqrt{2}n^{1/6}},~t=\sqrt{2n}+\frac{s}{\sqrt{2}n^{1/6}},$
and $\phi(z^*):=\lim\limits_{n\to\infty}\phi_{n}(z)$
\begin{equation}\label{inftyphiGUE}
\begin{aligned}
\phi''(z^*)&+\left(\frac{1}{z^*-s}-\frac{1}{z^*-s+u(s)}\right)\phi'(z^*)\\
&+\left(\frac{\frac{\left(u'(s)\right)^2}{4u(s)}-\frac{u'(s)}{2u(s)}-su(s)+u^2(s)}{z^*-s}+\frac{\frac{u'(s)}{2u(s)}}{z^*-s+u(s)}-z^*\right)\phi(z^*)=0,
\end{aligned}
\end{equation}
where $u(s):=\lim\limits_{n\to\infty}\sqrt{2}n^{1/6}\alpha_{n}(t)$ is a solution of \eqref{limitrn}.
\end{enumerate}
\end{theorem}
\begin{proof}
\begin{itemize}
\item[(i)]
Substitution of \eqref{GUEAnBn} in (\ref{odephin}) leads to
\begin{equation}\label{GUEphin1}
\begin{aligned}
\phi_{n}''(z)&+\left(\frac{1}{z-t}-\frac{1}{z-t+\alpha_n(t)}\right)\phi_{n}'(z)\\
&+\left(\frac{-\frac{r_n(t)}{\alpha_n(t)}+t+2\sum\limits^{n-1}_{j=0}\alpha_j(t)}{z-t}+\frac{\frac{r_n(t)}{\alpha_n(t)}-t+\alpha_n(t)}{z-t+\alpha_n(t)}+2n+1-z^2\right)\phi_{n}(z)=0.
\end{aligned}
\end{equation}
The following results are established in \cite{ChenPruessner2005} (see (22), (24), (26), (28) and (31)),
\begin{subequations}
\begin{align}
r_n^2(t)&=2(n+r_n(t))\alpha_n(t)\alpha_{n-1}(t),\label{rnalphan-1}\\
r_n'(t)&=2(n+r_n(t))(\alpha_{n-1}(t)-\alpha_n(t)),\label{Drnalphan-1}\\
r_n(t)&=\alpha_n(t)(t-\alpha_n(t))+\frac{1}{2}\alpha_n'(t),\label{rnDalphan}\\
-2\sum\limits^{n-1}_{j=0}\alpha_j(t)&=2tr_n(t)-2(n+r_n(t))(\alpha_n(t)+\alpha_{n-1}(t)).\label{sumalpha}
\end{align}
\end{subequations}
To remove $\alpha_{n-1}(t)$ from \eqref{sumalpha}, we add \eqref{Drnalphan-1} to it and get
$$2\sum\limits^{n-1}_{j=0}\alpha_j(t)=r_n'(t)+(4\alpha_n(t)-2t)r_n(t)+4n\alpha_n(t),$$
so that on account of \eqref{rnDalphan}, equation \eqref{odephialphan} follows from \eqref{GUEphin1}.
Equation \eqref{GUEalphan} is true due to \eqref{rnalphan-1}-\eqref{rnDalphan}.
In fact, from \eqref{rnalphan-1} and \eqref{Drnalphan-1}, there follows
$$\frac{r_n^2(t)}{\alpha_n(t)}=2(n+r_n(t))\alpha_{n-1}(t)=r_n'(t)+2(n+r_n(t))\alpha_n(t),$$
which combined with \eqref{rnDalphan} implies \eqref{GUEalphan}.
\item[(ii)]
For $t=\sqrt{2n}+\frac{s}{\sqrt{2}n^{1/6}}$ and $u(s)=\lim\limits_{n\to\infty}\sqrt{2}n^{1/6}\alpha_{n}(t),$
we get
$$u'(s)=\lim\limits_{n\to\infty}\alpha_{n}'(t),
\qquad u''(s)=\lim\limits_{n\to\infty}\frac{\alpha_{n}''(t)}{\sqrt{2}n^{1/6}}.$$
If we divide both sides of \eqref{GUEalphan} by $\sqrt{2}n^{1/6}$ and taking $n\rightarrow\infty,$
then we obtain
\bean
u''(s)=\frac{1}{2}\cdot\frac{u'(s)^2}{u(s)}+2su(s)-4u^2(s).
\eean
In addition, according to \eqref{finiteGUEpara},
\bean
\lim\limits_{n\to\infty}\frac{\kappa_3(t)}{\sqrt{2}n^{1/6}}=\frac{\left(u'(s)\right)^2}{4u(s)}-\frac{u'(s)}{2u(s)}-su(s)+u^2(s),
\qquad \lim\limits_{n\to\infty}\frac{\kappa_4(t)}{\sqrt{2}n^{1/6}}=\frac{u'(s)}{2u(s)}.
\eean
For $z=\sqrt{2n}+\frac{z^*}{\sqrt{2}n^{1/6}}$ and $\phi(z^*)=\lim\limits_{n\to\infty}\phi_{n}(z),$ we have
$$\phi'(z^*)=\lim\limits_{n\to\infty}\frac{\phi_{n}'(z)}{\sqrt{2}n^{1/6}},\qquad\phi''(z^*)=\lim\limits_{n\to\infty}\frac{\phi_{n}''(z)}{2n^{1/3}}.$$
Dividing both sides of \eqref{GUEphin} by $2n^{1/3},$
as $n\rightarrow\infty,$ (\ref{inftyphiGUE}) is seen to be true.
\end{itemize}
\end{proof}

\begin{remark}
Note that \eqref{inftyphiGUE} is identical with \eqref{inftyphi}. This result is mainly due to the relation between Hermite and Laguerre polynomials.
Indeed, monic Hermite polynomials $\{H_n(z)\}$ can be reduced to monic Laguerre polynomials $\{L_n^\alpha(z)\}$ by
$$H_{2n}(z)=L_n^{\left(-\frac{1}{2}\right)}(z^2),\qquad H_{2n+1}(z)=L_n^{\left(\frac{1}{2}\right)}(z^2).$$
Observe that $\alpha=\pm\frac{1}{2}$ corresponds to $\mu=0$ (see Section \ref{largen}). Write
\bean
z_L&:=&4n+4^{2/3}n^{1/3}z^*,\qquad t_L:=4n+4^{2/3}n^{1/3}s,\\
z&:=&\sqrt{2n}+\frac{z^*}{\sqrt{2}n^{1/6}},\qquad t:=\sqrt{2n}+\frac{s}{\sqrt{2}n^{1/6}}.
\eean
Replacing $n$ by $2n$ in $z$ and $t,$ we find
\bean
z^2&=&4n+4^{2/3}n^{1/3}z^*+\frac{(z^*)^2}{4^{2/3}n^{1/3}}\sim z_L,\\
t^2&=&4n+4^{2/3}n^{1/3}s+\frac{s^2}{4^{2/3}n^{1/3}}\sim t_L,
\eean
Here the symbol $\sim$ refers to the limiting procedure $n\rightarrow\infty.$
The above analysis suggests that \eqref{inftyphi} and \eqref{inftyphiGUE} are obtained by using the same scaling method.
\end{remark}

Introducing into (\ref{inftyphi}) the new variable $x$ defined by $x=-\frac{z^*-s}{r(s)}$
and the new function $f(x)=\phi(z^*),$ we obtain an equation in the form
\bea\label{inftyphipq}
f''(x)+p(x)f'(x)+q(x)f(x)=0,
\eea
where
\bean
p(x):=\frac{1}{x}-\frac{1}{x-1},\qquad q(x):=\frac{a_0}{x}+\frac{a_1}{x-1}+a_2+a_3x,
\eean
and
\bean
&a_0=-\frac{1}{4}r'(s)^2+\frac{1}{2}r'(s)+sr^2(s)-r^3(s),\\
&a_1=-\frac{1}{2}r'(s),\qquad a_2=-sr^2(s),\qquad a_3=r^3(s).
\eean
Note that the dependence on $s$ of $\{a_i\}$ is not displayed for ease of notations,
in addition, the Painlev$\acute{e}$ equation
$$a_0=-a_1^2-a_1-a_2-a_3.$$
For any interval $[c_0,x]$ excluding $0,1,$ and $\infty,$ by writing
$$f(x)=F(x)\exp\left(-\frac{1}{2}\int_{c_0}^x p(z)dz\right),$$
it follows from \eqref{inftyphipq} that
\bea\label{nofirstderi}
F''(x)+J(x)F(x)=0,
\eea
where
\bean
J(x)=\frac{1}{4x^2}+\frac{a_0-\frac{1}{2}}{x}-\frac{3}{4(x-1)^2}+\frac{a_1+\frac{1}{2}}{x-1}
+a_2+a_3x.
\eean
Since
\begin{equation}
\begin{aligned}
J(x)\cong
\begin{cases}
    \frac{1}{4x^2}+\frac{a_0-\frac{1}{2}}{x}+a_2-a_1-\frac{5}{4},       & \qquad\text{as} \quad x\rightarrow0,\\
    -\frac{3}{4(x-1)^2}+\frac{a_1+\frac{1}{2}}{x-1}-a_1^2-a_1-\frac{1}{4},  & \qquad \text{as } \quad x\rightarrow1,\\
    a_2+a_3 x,& \qquad \text{as } \quad x\rightarrow\infty,
\end{cases}
\end{aligned}
\end{equation}
according to \eqref{nofirstderi}, we have the following asymptotic formulas
\begin{equation}
\begin{aligned}
F(x)\cong
\begin{cases}
C_1\frac{\sqrt{2\lambda x}}{{\exp\left(\lambda x\right)}}
{{\rm M}\left(\frac{1}{2}-{\frac {a_0-\frac{1}{2}}{2\lambda }},1,2\lambda x\right)}+C_2\frac{\sqrt{2\lambda x}}{{\exp\left(\lambda x\right)}}
{{\rm U}\left(\frac{1}{2}-{\frac {a_0-\frac{1}{2}}{2\lambda }},1,2\lambda x\right)},       & \qquad\text{as} \quad x\rightarrow0,\bigskip\\
C_3{\frac {{\exp\left(\left(a_{{1}}+\frac{1}{2}\right)x\right)}}{
\sqrt {x-1}}}+C_4{\frac {\left(\left(a_{{1}}+\frac{1}{2}\right)x-a_1\right) \exp\left(-\left(a_{{1}}+\frac{1}{2}\right)x\right)}{\sqrt {x-1}}},  & \qquad \text{as} \quad x\rightarrow1,\bigskip\\
C_5{{\rm Ai}\left(-\frac{a_2}{\sqrt[3]{a_3^2}}-\sqrt[3]{a_3}x\right)}+C_{6}
{{\rm Bi}\left(-\frac{a_2}{\sqrt[3]{a_3^2}}-\sqrt[3]{a_3}x\right)},& \qquad \text{as} \quad x\rightarrow\infty,\non
\end{cases}
\end{aligned}
\end{equation}
where $\lambda :=\frac{1}{2}\sqrt{-4a_{{2}}+4a_{{1}}+5},$
${\rm M}(\mu,\nu,z)$ and ${\rm U}(\mu,\nu,z)$ are Kummer's confluent hypergeometric functions M and U respectively,
${\rm Ai}(z)$ and ${\rm Bi}(z)$ are the Airy functions of the first and second kind respectively, and $C_1$-$C_6$ are arbitrary constants.

We readily see that
$x=0$ is a regular singular point \cite{WhittakerWatson1927} for \eqref{inftyphipq} since $xp(x)$ and $x^2q(x)$ are analytic at 0.
From \eqref{inftyphipq} it follows that
\bea\label{eqn0}
D(f):=(x^2-x)f''(x)-f'(x)+\left(a_3x^3 +b_2x^2+ b_1x-a_0\right)f(x)=0,
\eea
where
$$
b_1=-\frac{1}{4}r'(s)^2+2sr^2(s)-r^3(s),\qquad b_2=-sr^2(s)-r^3(s).
$$
Note that
$$a_0+a_1-a_3-b_1-b_2=0.$$
Let $$Y_0=\sum\limits^{\infty}_{n=0}c_{0,n} x^{\tau+n},\qquad c_{0,0}=1,
$$
then
\bean
-xD(Y_0)=\tau^2x^\tau,
\eean
provided that for $n\geq0$
\bea\label{c0fiverecurrencetau}
\qquad (\tau+n+1)^2 c_{0,n+1}-((\tau+n-1)(\tau+n)-a_0)c_{0,n}-b_1 c_{0,n-1}-b_2 c_{0,n-2}-a_3 c_{0,n-3}=0.\quad
\eea
Here and in what follows $c_{0,-3}=c_{0,-2}=c_{0,-1}:=0.$
The two solutions of \eqref{eqn0} are given by
$$\left.Y_0\right|_{\tau=0}=\sum\limits^{\infty}_{n=0}c_{0,n} x^{n},
\qquad \left.\frac{\partial Y_0}{\partial \tau}\right|_{\tau=0}=\sum\limits^{\infty}_{n=1}d_{0,n} x^{n}+\ln x \cdot \left.Y_0\right|_{\tau=0}$$
where
$$d_{0,n}=\left.\frac{\partial c_{0,n}}{\partial \tau}\right|_{\tau=0}.$$
See \cite{Forsyth1956}.
Using relation \eqref{c0fiverecurrencetau}, we find $\{c_{0,n}\}$ with $n$ large appear in the following form
\bean
c_{0,n}=\sum\limits^{\infty}_{\ell=3}\frac{\theta_{0,\ell}}{n^{\ell}},\qquad\theta_{0,3}:=1,
\eean
where the coefficients $\{\theta_{0,\ell}\}_{\ell\geq1}$ are determined by
\begin{equation}\label{largencntau}
\begin{aligned}
\ell\theta_{0,\ell+3}=&\left(-\tau(1+2\ell)+\frac{1}{2}(\ell+1)\ell-a_1\right)\theta_{0,\ell+2}\\
&+\sum\limits^{\ell+1}_{k=3}\left\{(-1)^{\ell-k}\left[\binom{\ell+1}{k-3}-2\tau\binom{\ell+1}{k-2}+\tau^2\binom{\ell+1}{k-1}\right]\right.\\
&\left.\qquad\qquad-\left(b_1+2^{\ell-k+2}b_2+3^{\ell-k+2}a_3\right)\binom{\ell+1}{k-1}\right\}\theta_{0,k}.
\end{aligned}
\end{equation}
For $\ell=1$ the sum term is to be replaced by 0, so that
$$\theta_{0,4}=-3\tau+1-a_1.$$
In case $\tau=0,$ we obtain from \eqref{c0fiverecurrencetau}
\bean
\qquad (n+1)^2 c_{0,n+1}-(n(n-1)-a_0)c_{0,n}-b_1 c_{0,n-1}-b_2 c_{0,n-2}-a_3 c_{0,n-3}=0,\qquad n\geq0,
\eean
and from \eqref{largencntau}
\begin{equation}
\begin{aligned}
\ell\theta_{0,\ell+3}=&\left(\frac{1}{2}(\ell+1)\ell-a_1\right)\theta_{0,\ell+2}\\
&+\sum\limits^{\ell+1}_{k=3}\left\{(-1)^{\ell-k}\binom{\ell+1}{k-3}-\left(b_1+2^{\ell-k+2}b_2+3^{\ell-k+2}a_3\right)\binom{\ell+1}{k-1}\right\}\theta_{0,k},\qquad \ell\geq1.
\end{aligned}
\end{equation}
Differentiation of \eqref{c0fiverecurrencetau} and \eqref{largencntau} yields, respectively,
\bean
&&\frac{\partial c_{0,n+1}}{\partial \tau}(\tau+n+1)^2+2c_{0,n+1}(\tau+n+1)-c_{0,n}(2\tau+2n-1)\\
&&\qquad-\frac{\partial c_{0,n}}{\partial \tau}((\tau+n-1)(\tau+n)-a_0)-\frac{\partial c_{0,n-1}}{ \partial \tau}b_1-\frac{\partial c_{0,n-2}}{ \partial \tau}b_2-\frac{\partial c_{0,n-3}}{ \partial \tau}a_3=0,\qquad n\geq0,
\eean
and
\bean
\ell\frac{\partial\theta_{0,\ell+3}}{\partial\tau}&=&-(1+2\ell)\theta_{0,\ell+2}+\sum\limits^{\ell+1}_{k=3}(-1)^{\ell-k}\left[-2\binom{\ell+1}{k-2}+2\tau\binom{\ell+1}{k-1}\right]\theta_{0,k}\\
&&+\left(-\tau(1+2\ell)+\frac{1}{2}(\ell+1)\ell-a_1\right)\frac{\partial\theta_{0,\ell+2}}{\partial\tau}\\
&&+\sum\limits^{\ell+1}_{k=3}\left\{(-1)^{\ell-k}\left[\binom{\ell+1}{k-3}-2\tau\binom{\ell+1}{k-2}+\tau^2\binom{\ell+1}{k-1}\right]\right.\\
&&\left.\qquad\qquad-\left(b_1+2^{\ell-k+2}b_2+3^{\ell-k+2}a_3\right)\binom{\ell+1}{k-1}\right\}\frac{\partial\theta_{0,k}}{\partial\tau},\qquad \ell\geq1.
\eean
Setting $\tau=0$ in the above two formulas, we find
\bean
(n+1)^2 d_{0,n+1}&+&2(n+1)c_{0,n+1}-(2n-1)c_{0,n}\\
&-&(n(n-1)-a_0)d_{0,n}-b_1 d_{0,n-1}-b_2 d_{0,n-2}-a_3 d_{0,n-3}=0,\qquad n\geq0,
\eean
and, by denoting $\nu_{0,j}=\left.\frac{\partial\theta_{0,j}}{\partial \tau}\right|_{\tau=0},$
\begin{equation}\label{largecn'tau}
\begin{aligned}
\ell\nu_{0,\ell+3}=&-(1+2\ell)\theta_{0,\ell+2}-2\sum\limits^{\ell+1}_{k=3}(-1)^{\ell-k}\binom{\ell+1}{k-2}\theta_{0,k}+\left(\frac{1}{2}(\ell+1)\ell-a_1\right)\nu_{0,\ell+2},\\
&+\sum\limits^{\ell+1}_{k=3}\left\{(-1)^{\ell-k}\binom{\ell+1}{k-3}-\left(b_1+2^{\ell-k+2}b_2+3^{\ell-k+2}a_3\right)\binom{\ell+1}{k-1}\right\}\nu_{0,k},\qquad \ell\geq1.
\end{aligned}
\end{equation}
Observe that for large $n$
\bean
d_{0,n}=\sum\limits^{\infty}_{\ell=4}\frac{\nu_{0,\ell}}{n^\ell},\qquad\nu_{0,4}=-3.
\eean

Now we proceed to find out the solutions of \eqref{inftyphipq} in the form of Taylor series near $x=1$
which is a regular singular point since $(x-1)p(x)$ and $(x-1)^2q(x)$ are analytic at 0.
Equation \eqref{eqn0} can be rewritten as
\begin{equation}\label{x=1equation}
\begin{aligned}
D(f)&:=\left((x-1)^2+(x-1)\right)f''(x)-f'(x)+\left(a_3(x-1)^3+A(x-1)^2-a_1^2(x-1)+a_1\right)f(x)\\
&=0,
\end{aligned}
\end{equation}
where
$$A=-sr^2(s)+2r^3(s).$$
Suppose
\bean
Y_1(x)=\sum\limits^{\infty}_{n=0}c_{1,n} (x-1)^{\lambda+n},
\eean
then
\bean
(x-1)D(Y_1)=c_{1,0}\lambda(\lambda-2)(x-1)^{\lambda},
\eean
providing that
\begin{equation}\label{z=1c1c2c3}
\begin{aligned}
\left(\lambda^2-1\right)c_{1,1}+(a_1+\lambda(\lambda-1))c_{1,0}=0&,\\
\lambda(\lambda+2)c_{1,2}+(a_1+\lambda(\lambda+1))c_{1,1}-a_1^2c_{1,0}=0&,\\
(\lambda+1)(\lambda+3)c_{1,3}+(a_1+(\lambda+1)(\lambda+2))c_{1,2}-a_1^2c_{1,1}+Ac_{1,0}=0&,
\end{aligned}
\end{equation}
and for $n\geq3$
\bea\label{z=1cnlamda}
(\lambda+n-1)(\lambda+n+1)c_{1,n+1}+(a_1+(\lambda+n-1)(\lambda+n))c_{1,n}-a_1^2c_{1,n-1}+Ac_{1,n-2}+a_3c_{1,n-3}=0.\non\\
\eea
%Here and in what follows $c_{1,-3}=c_{1,-2}=c_{1,-1}:=0.$
Taking $c_{1,0}=\frac{1}{2}\lambda,$ we have
\bean
(x-1)D(Y_1)=\frac{1}{2}\lambda^2(\lambda-2)(x-1)^{\lambda},
\eean
so that there are three solutions of \eqref{x=1equation} given by
$$\left.Y_1\right|_{\lambda=2},
\qquad \left.Y_1\right|_{\lambda=0},\qquad\left.\frac{\partial Y_1}{\partial \lambda}\right|_{\lambda=0}.$$
Given that $c_{1,0}=\frac{1}{2}\lambda,$ we obtain from \eqref{z=1c1c2c3}
\bean
\left.c_{1,n}\right|_{\lambda=0}=0,\qquad n=0,1,2,3,
\eean
so that by mathematical induction it follows from \eqref{z=1cnlamda} that
$$\left.c_{1,n}\right|_{\lambda=0}=0,\qquad n\geq0,$$
which implies
$$
\left.Y_1\right|_{\lambda=0}=0,\qquad
\left.\frac{\partial Y_1}{\partial \lambda}\right|_{\lambda=0}
=\sum\limits^{\infty}_{n=0}\left.\frac{\partial c_{1,n}}{\partial \lambda}\right|_{\lambda=0} (x-1)^{n}.
$$
Consequently, there are two solutions of \eqref{x=1equation} given by
$$\left.Y_1\right|_{\lambda=2}=\sum\limits^{\infty}_{n=0}\left(\left.c_{1,n}\right|_{\lambda=2}\right) (x-1)^{n+2},
\qquad 2\left.\frac{\partial Y_1}{\partial \lambda}\right|_{\lambda=0}=\sum\limits^{\infty}_{n=0}d_{1,n} (x-1)^{n},$$
where
$$d_{1,n}=2\left.\frac{\partial c_{1,n}}{\partial \lambda}\right|_{\lambda=0}.$$
Differentiation of \eqref{z=1cnlamda} gives
\begin{equation}\label{z=1cn'}
\begin{aligned}
&2(\lambda+n)c_{1,n+1}+(2(\lambda+n)-1)c_{1,n}+\left((\lambda+n)^2-1\right)\frac{\partial c_{1,n+1}}{\partial\lambda}\\
&+(a_1+(\lambda+n-1)(\lambda+n))\frac{\partial c_{1,n}}{\partial \lambda}-a_1^2\frac{\partial c_{1,n-1}}{\partial\lambda}
+A\frac{\partial c_{1,n-2}}{\partial\lambda}+a_3\frac{\partial c_{1,n-3}}{\partial\lambda}=0.
\end{aligned}
\end{equation}
In particular, for $\lambda=0:$
\bea
\left(n^2-1\right)d_{1,n+1}+(a_1+n(n-1))d_{1,n}-a_1^2d_{1,n-1}+Ad_{1,n-2}+a_3d_{1,n-3}=0.
\eea
Choosing $\lambda=2$ in \eqref{z=1cnlamda} leads to
\bean
(n+1)(n+3)c_{1,n+1}+(a_1+(n+1)(n+2))c_{1,n}-a_1^2c_{1,n-1}+Ac_{1,n-2}+a_3c_{1,n-3}=0.
\eean
When $n$ is large, we find from \eqref{z=1cnlamda}
\bean
c_{1,n}&=&(-1)^n\sum\limits^{\infty}_{\ell=1}\frac{\theta_{1,\ell}}{n^\ell},
\eean
where
\bean
\theta_{1,1}&:=&\frac{\lambda}{2},\non\\
\theta_{1,2}&=&\left(a_0-\lambda(\lambda-1)\right)\theta_{1,1},\non\\
2\theta_{1,3}&=&(a_0-3\lambda+2)\theta_{1,2}-(a_1^2+a_2+a_3)\theta_{1,1},\non\\
3\theta_{1,4}&=&\left(a_0+5(1-\lambda)\right)\theta_{1,3}-2\left(a_1^2+a_2+a_3+\lambda^2-3\lambda+1\right)\theta_{1,2}-(a_1^2+4a_2-a_3)\theta_{1,1},\\
\ell\theta_{1,\ell+1}&=&\left(a_0+\frac{\ell(\ell+1)}{2}-(2\ell-1)\lambda-1\right)\theta_{1,\ell}\\
&&+\sum\limits^{\ell-3}_{k=1}\left\{(-1)^{\ell-k}\left[\binom{\ell}{k}+(1-2\lambda)\binom{\ell-1}{k}+\lambda(\lambda-2)\binom{\ell-1}{k+1}\right]\right.\\
&&\left.\qquad\qquad-\binom{\ell-1}{k+1}\left(a_1^2+2^{\ell-k-2}A-3^{\ell-k-2}a_3\right)\right\}\theta_{1,k+2}\\
&&-(\ell-1)\left(a_1^2+2^{\ell-2}A-3^{\ell-2}a_3\right)\theta_{1,2}-(a_1^2+2^{\ell-1}A-3^{\ell-1}a_3)\theta_{1,1},\qquad\ell\geq4.
\eean
Recalling $d_{1,n}=2\left.\frac{\partial c_{1,n}}{\partial \lambda}\right|_{\lambda=0}$
and bearing in mind $\left.c_{1,n}\right|_{\lambda=0}=0$ which suggests $\left.\theta_{1,\ell}\right|_{\lambda=0}=0,$
we obtain for large $n,$
\bean
d_{1,n}&=&(-1)^n \sum\limits^{\infty}_{\ell=1}\frac{\nu_{1,\ell}}{n^\ell},
\eean
where
$
\nu_{1,\ell}=2\left.\frac{\partial \theta_{1,\ell}}{\partial \lambda}\right|_{\lambda=0}
$
and satisfy
\bean
\nu_{1,1}&=&1,\\
\nu_{1,2}&=&a_0,\non\\
2\nu_{1,3}&=&\left(a_0+2\right)\nu_{1,2}-(a_1^2+a_2+a_3),\\
3\nu_{1,4}&=&\left(a_0+5\right)\nu_{1,3}-2\left(a_1^2+a_2+a_3+1\right)\nu_{1,2}-(a_1^2+4a_2-a_3),\\
\ell\nu_{1,\ell+1}&=&\left(a_0+\frac{(\ell-1)(\ell+2)}{2}\right)\nu_{1,\ell}\\
&&+\sum\limits^{\ell-3}_{k=1}\left\{(-1)^{\ell-k}\left[\binom{\ell}{k}+\binom{\ell-1}{k}\right]-\binom{\ell-1}{k+1}\left(a_1^2+2^{\ell-k-2}A-3^{\ell-k-2}a_3\right)\right\}\nu_{1,k+2}\\
&&-(\ell-1)\left(a_1^2+2^{\ell-2}A-3^{\ell-2}a_3\right)\nu_{1,2}-(a_1^2+2^{\ell-1}A-3^{\ell-1}a_3),\qquad\ell\geq4.
\eean
In addition, for $\lambda=2,$
\bean
\theta_{1,1}&=&1,\non\\
\theta_{1,2}&=&a_0-2,\non\\
2\theta_{1,3}&=&\left(a_0-4\right)\theta_{1,2}-(a_1^2+a_2+a_3)\\
3\theta_{1,4}&=&\left(a_0-5\right)\theta_{1,3}-2\left(a_1^2+a_2+a_3-1\right)\theta_{1,2}-(a_1^2+4a_2-a_3),\\
\ell\theta_{1,\ell+1}&=&\left(a_0+\frac{\ell(\ell+1)}{2}-4\ell+1\right)\theta_{1,\ell}\\
&&+\sum\limits^{\ell-3}_{k=1}\left\{(-1)^{\ell-k}\left[\binom{\ell}{k}-3\binom{\ell-1}{k}\right]-\binom{\ell-1}{k+1}\left(a_1^2+2^{\ell-k-2}A-3^{\ell-k-2}a_3\right)\right\}\theta_{1,k+2}\\
&&-(\ell-1)\left(a_1^2+2^{\ell-2}A-3^{\ell-2}a_3\right)\theta_{1,2}-(a_1^2+2^{\ell-1}A-3^{\ell-1}a_3),\qquad\ell\geq4.
\eean

\section{Chazy's equations}
%Noticing that $r_n(t)=\sigma_n'(t),$ we can find the ODE for $r_n(t)$ from the $\sigma$-form of a particular $P_V$ satisfied by $\sigma_n(t).$
%In general, we can develop the ODE for $\varrho(t):=\Xi'(t)$ with $\Xi(t)$ satisfying the $\sigma$-form of a general $P_{V}.$
Noticing that $r_n(t)=\sigma_n'(t)$ is valid for both our problem and the largest eigenvalue distribution of GUE (see \cite{ChenPruessner2005}),
we can find the ODE for $r_n(t)$ from the $\sigma$-form of a particular $P_V$ satisfied by $\sigma_n(t).$
In general, we can develop the ODE for $\varrho(t):=\Xi'(t)$
with $\Xi(t)$ satisfying the $\sigma$-form of a general $P_{V}.$
The extension of these considerations to $P_{IV}$ is not difficult.

\begin{subsection}{Chazy's equation for the $\sigma$-form of $P_V$}
Recall the J-M-O $\sigma$-form of $P_V$ \cite{Jimbo1981}
\bea\label{JMOsigmaPV}
(t\Xi '')^2=\left(\Xi-t \Xi
   '+2 (\Xi ')^2+\left(\nu _1+\nu _2+\nu _3\right) \Xi '\right){}^2-4 \Xi ' \left(\Xi '+\nu _1\right) \left(\Xi
   '+\nu _2\right) \left(\Xi '+\nu _3\right),
\eea
where $\Xi=\Xi(t)$  and $\nu_1,~\nu_2,~\nu_3$ are constants. Here we use the new notation $\Xi$ instead of $\sigma.$
To get the ODE for $\varrho(t):=\Xi'(t),$ we need to eliminate $\Xi$ in (\ref{JMOsigmaPV}).
Differentiating both sides of (\ref{JMOsigmaPV}),
solving for $\Xi(t)$ and plugging it back into (\ref{JMOsigmaPV}), we obtain
\bea\label{oderngeneral}
&&\left(t \left(\varrho'+t \varrho''\right)+8
   \varrho^3+6 \left(\nu _1+\nu
   _2+\nu _3\right) \varrho^2+4 \left(\nu_1\nu_2+\nu_1\nu_3+\nu _2 \nu _3\right) \varrho+2 \nu _1 \nu _2 \nu _3\right){}^2\nonumber\\
&=&\left(4 \varrho+\nu _1+\nu _2+\nu _3-t\right){}^2 \nonumber\\
&&\cdot\left(t^2 (\varrho')^2+4 \varrho^4+4 \left(\nu _1+\nu _2+\nu _3\right)
   \varrho^3+4 \left(\nu_1\nu_2+\nu_1\nu_3+\nu _2 \nu _3\right) \varrho^2+4 \nu _1 \nu _2 \nu _3\varrho\right).
\eea
With the change of variables,
$$
t = 2 i e^z,\qquad
\vartheta(z)=-2i\varrho(t)-\frac{i}{2} \left(\nu_1+\nu _2+\nu_3\right),\nonumber
$$
where $i^{2}=-1,$ we see that $\vartheta(z)$ satisfies
\begin{subequations}\label{ourChazyLUE}
\begin{align}\label{ChazyII2equation}
\left(\frac{d^2 \vartheta}{d z^2}-2 \vartheta^3-\alpha_{1} \vartheta-\beta_{1}\right)^2=-4
   \left(\vartheta-e^z\right)^2 \left(\left(\frac{d \vartheta}{dz}\right)^2-\vartheta^4-\alpha_{1} \vartheta^2-2\beta_{1}\vartheta-\gamma_{1}\right),
\end{align}
recognized to be the second member of the Chazy II system \cite{Cosgrove2006} with
\begin{equation}\label{ChazyII2parameters}
\begin{aligned}
\alpha_1&=\frac{1}{2} \left(3 \nu _1^2+3 \nu _2^2+3 \nu _3^2-2 \nu _1\nu _2-2 \nu _1\nu _3-2 \nu _2 \nu _3\right),\\
\beta_1&=\frac{i}{2} \left(\nu _1-\nu _2-\nu _3\right) \left(\nu _1+\nu _2-\nu _3\right) \left(\nu _1-\nu _2+\nu
   _3\right),\\
\gamma_1&=-\frac{1}{16} \left(\nu _1+\nu _2-3 \nu _3\right) \left(3 \nu _1-\nu _2-\nu _3\right) \left(\nu _1-3 \nu
   _2+\nu _3\right) \left(\nu _1+\nu _2+\nu _3\right).
\end{aligned}
\end{equation}
\end{subequations}

Now we proceed to apply our results to different ensembles.

\begin{example}
For our problem, we have
$$\Xi(t)=\sigma_n(t),\qquad\varrho(t)=r_n(t).$$
Substituting (\ref{paralargest}) for $\nu_1,~\nu_2$ and $\nu_3$ in (\ref{oderngeneral}) and (\ref{ourChazyLUE})
gives rise to, respectively, the ODE for $r_n(t)$
\bea\label{odern}
&&\left(t \left(r_n'+t r_n''\right)+8
   r_n^3+6 \left(2n+\alpha\right) r_n^2+4n(n+\alpha) r_n\right){}^2\nonumber\\
&&=\left(4 r_n+2n+\alpha-t\right){}^2\left(t^2 (r_n')^2+4 r_n^4+4 \left(2n+\alpha\right)
   r_n^3+4 n(n+\alpha) r_n^2\right),
\eea
and the Chazy equation (\ref{ChazyII2equation}) for
$$\vartheta(z)=-2i\left.r_n(t)\right|_{t = 2 i e^z}-\frac{i}{2} \left(2n+\alpha\right)$$
with
\begin{equation}\label{ourparameters}
\begin{aligned}
\alpha_{1}=\frac{1}{2} \left(3 \alpha ^2+4 n^2+4 \alpha n\right),
\qquad\beta_{1}=\frac{i}{2} \alpha ^2 (\alpha +2 n),
\qquad\gamma_{1}=\frac{1}{16} (2n-\alpha) (\alpha +2 n)^2 (3\alpha +2 n).\non
\end{aligned}
\end{equation}
\begin{remark}\label{secondChazyderi}
By means of the Riccati equations for $R_n$ and $r_n,$ we can derive the ODE for $r_n$ in an alternative way which is straightforward but tedious.
Solving for $R_{n}$ from (\ref{Riccattirn}) yields
\bean
R_n=\frac{t  r_n'+2
   r_n^2\pm\sqrt{\Delta_n}}{2 \left(t
   r_n'-(2n+\alpha)  r_n -n(n+\alpha)\right)},
\eean
with $\Delta_n:=4 n (\alpha +n) r_n^2+4 (\alpha +2 n) r_n^3+4 r_n^4+(t r_n')^2.$
Choosing either sign, substituting the resulting $R_n$ into (\ref{RiccatiRn})
and clearing the square root, we obtain (\ref{odern}) again.
\end{remark}
\end{example}

\begin{example}
The outage probability of a single-user MIMO communication system can be calculated via
the moment generating function $\mathcal{M}(\lambda)$ of the mutual information.
Under some assumptions, see \cite{BasorChen2014} and \cite{ChenMcKay2012}, $\mathcal{M}(\lambda)$ is shown to be
$$\mathcal{M}(\lambda)=t^{-n\lambda}\frac{D_n(t,\lambda)}{D_n(t,0)}=\exp\left(\int_\infty^t\frac{H_n(x)-n\lambda}{x}\right)dx,$$
where
\bean
D_n(t,\lambda)&=&\det\left(\int_0^\infty x^{i+j}w_{\rm dlag}(x,t)dx\right)_{i,j=0}^{n-1},\non\\
H_n(t)&:=&t\frac{d}{dt}\ln D_n(t,\lambda),\non
\eean
and
$$w_{\rm dlag}(x,t)=(x+t)^\lambda x^\alpha e^{-x}.$$
Moreover, the $\sigma$-form of $P_V$ with
\bea\label{ex1P5parameters}
\nu_1=\lambda,\qquad\nu_2=-n,\qquad\nu_3=-n-\alpha
\eea
is established for
$$\Xi(t):=H_n(t)-n\lambda,$$
and the following is established
$$r_n(t):=\frac{\lambda}{h_{n-1}}\int_0^\infty \frac{P_n(x)P_{n-1}(x)}{x+t}w_{\rm dlag}(x,t)dx=-\Xi'(t),$$
where
\bea
\int_{0}^{\infty}P_m(x)P_{n}(x)w_{\rm dlag}(x,t)dx= h_n\delta_{mn},\qquad m\geq0,\quad n\geq0.\non
\eea
Plugging (\ref{ex1P5parameters}) into (\ref{oderngeneral}) leads to the ODE for
$$\varrho(t)=-r_n(t),$$
from which follows the Chazy's equation (\ref{ChazyII2equation}) for
$$\vartheta(z)=2i\left.r_n(t)\right|_{t = 2 i e^z}+\frac{i}{2} \left(2n+\alpha-\lambda\right)$$
with
\begin{equation}
\begin{split}
\alpha _1&=\frac{1}{2} \left(4 n^2+4 n \alpha+3 \alpha ^2+4n\lambda+2 \alpha  \lambda +3\lambda ^2\right),\notag\\
\beta _1&=-\frac{i}{2} (\alpha-\lambda )(\alpha+\lambda) (2 n+\alpha+\lambda ),\notag\\
\gamma _1&=\frac{1}{16} (2n+\alpha -\lambda) (2n-\alpha +\lambda) (2n+3\alpha +\lambda) (2n+\alpha+3 \lambda).\notag
\end{split}
\end{equation}
\end{example}

\begin{example} Denote by $D_n(t)$ the determinant of the Hankel matrix generated from the moments of the time-dependent Jacobi weight
$$w(x,t)=(1-x)^\alpha(1+x)^\beta e^{-tx},\qquad -1\leq x\leq1,\quad t\in\mathbb{R},$$
namely,
$$D_n(t)=\det\left(\int_{-1}^1 x^{i+j}w(x,t)dx\right)_{i,j=0}^{n-1}.$$
It is proved in \cite{BasorChenEhrhardt2010} that
$$\Xi(t):=t\frac{d}{dt}\ln D_n(t/2)-\frac{nt}{2}+n(n+\beta)$$
satisfies the $\sigma$-form of $P_V$ with
\bea\label{PVtimedependentcoef}
\nu_1=-\alpha,\qquad\nu_2=n,\qquad\nu_3=n+\beta.
\eea
Furthemore,
$$\Xi'(t)=-r_n(t/2),$$
where $r_n(t)$ is defined by
$$r_n(t):=\frac{\alpha}{h_{n-1}}\int_{-1}^1 \frac{P_n(x)P_{n-1}(x)}{1-x}w(x,t)dx,$$
with
\bea
\int_{-1}^{1}P_m(x)P_{n}(x)w(x,t)dx= h_n\delta_{mn},\qquad m\geq0,\quad n\geq0.\non
\eea
Substitution of (\ref{PVtimedependentcoef}) in (\ref{oderngeneral}) gives the ODE for
$$\varrho(t)=-r_n(t/2),$$
from which we obtain (\ref{ChazyII2equation}) for
$$\vartheta(z)=2i\left.r_n(t)\right|_{t = i e^z}-\frac{i}{2} \left(2n-\alpha+\beta\right)$$
with
\begin{equation}
\begin{split}
\alpha _1&=\frac{1}{2} \left(4 n^2+4 n \alpha+3 \alpha ^2+4n\beta+2 \alpha\beta +3 \beta ^2\right),\\
\beta _1&=-\frac{i}{2} (\alpha -\beta ) (\alpha+\beta ) (2n+\alpha +\beta),\\
\gamma _1&=\frac{1}{16} (2n+\alpha -\beta) (2n-\alpha+\beta) (2n+3 \alpha +\beta)(2n+\alpha +3 \beta).\nonumber
\end{split}
\end{equation}
\end{example}
\begin{example}
The weight
$$w(x,t)=e^{-t/x}x^\alpha(1-x)^\beta,\qquad 0\leq x\leq1,\quad t\geq0$$
is studied in \cite{ChenDai2010} and the J-M-O $\sigma$-form of $P_V$ is found for
$$\Xi(t):=t\frac{d}{dt}\ln D_n(t)-n(n+\alpha+\beta)$$
with
\bea\label{PVJacobicoef}
\nu_1=-(n+\alpha+\beta),\qquad\nu_2=n,\qquad\nu_3=-\beta.
\eea
Here again $D_n(t)$ is the Hankel determinant
$$D_n(t)=\det\left(\int_{0}^1 x^{i+j}w(x,t)dx\right)_{i,j=0}^{n-1}.$$
Take into account the following result in \cite{ChenDai2010}
$$\Xi'(t)=(2n+\alpha+\beta)\frac{r_n^*(t)}{t}-n$$
where
$$r_n^*(t):=\frac{t}{h_{n-1}}\int_0^1\frac{P_n(x)P_{n-1}(x)}{x}w(x,t)dx$$
with
$$
\int_{0}^{1}P_m(x)P_{n}(x)w(x,t)dx= h_n\delta_{mn},\qquad m\geq0,\quad n\geq0,
$$
then, with the aid of (\ref{PVJacobicoef}), we find the ODE from (\ref{oderngeneral}) for
$$\varrho(t)=(2n+\alpha+\beta)\frac{r_n^*(t)}{t}-n,$$
so that (\ref{ChazyII2equation}) holds for
$$\vartheta(z)=-2i\left.\varrho(t)\right|_{t = 2 i e^z}+\frac{i}{2} \left(\alpha+2\beta\right)$$
with
\begin{equation}
\begin{split}
\alpha _1&=\frac{1}{2} \left(8 n^2+8 n \alpha+3 \alpha ^2+8 n \beta +4 \alpha  \beta +4\beta ^2\right),\notag\\
\beta _1&=-\frac{i}{2} \alpha  (2n+\alpha) (2n+\alpha +2 \beta),\notag\\
\gamma _1&=-\frac{1}{16} (\alpha -2 \beta )(\alpha +2 \beta ) (4n+\alpha +2\beta) (4 n +3 \alpha +2\beta).\notag
\end{split}
\end{equation}
\end{example}
\end{subsection}

\begin{subsection}{Chazy's equation for the $\sigma$-form of $P_{IV}$}

We go ahead to deal with the $\sigma$-form of $P_{IV}$ in the spirit of the previous section.
The J-M-O $\sigma$-form of $P_{IV}$ \cite{Jimbo1981} is given by
\bea
(\Xi '')^2 =4\left(t \Xi '-\Xi\right)^2 -4 \Xi ' \left(\Xi '+\nu_1\right)\left(\Xi'+\nu _2\right),\nonumber
\eea
where $\Xi=\Xi(t)$  and $\nu_1,~\nu_2$ are constants. From this equation follows the ODE for $\varrho(t):=\Xi'(t),$
\bea\label{generalChazy1}
\left(\varrho''+6 \varrho^2+4\left(\nu _1+\nu _2\right)\varrho+2\nu _1 \nu _2\right){}^2 =4 t^2 \left((\varrho')^2+4\varrho^3+4(\nu_1+\nu_2)\varrho^2+ 4\nu_1\nu_2\varrho\right).
\eea
Setting
$$
\vartheta(z)=-\frac{1}{2}\left.\varrho(t)\right|_{t=\frac{z}{\sqrt{2}}}-\frac{1}{6} \left(\nu _1+\nu_2\right),$$
%$$
%t=\frac{z}{\sqrt{2}},~v(z)=-\frac{1}{2}\varrho(t)-\frac{1}{6} \left(\nu _1+\nu_2\right),$$
%
%\bea
%t&=&\frac{z}{\sqrt{2}},\nonumber\\
%\varrho(t)&=&-2 v(z)-\frac{1}{3} \left(\nu _1+\nu_2\right),\nonumber
%\eea
we find
\begin{subequations}\label{ChazyII1}
\begin{equation}\label{generalChazyII1}
\left(\vartheta''(z)-6 \vartheta(z)^2-\alpha_1\right){}^2=z^2 \left( \vartheta'(z)^2-4 \vartheta(z)^3-2\alpha_1 \vartheta(z)-\beta_1\right),
\end{equation}
which is the first member of the Chazy II system \cite{Cosgrove2006} with
\begin{equation}\label{ourparameters1}
\begin{aligned}
\alpha_1=\frac{1}{6}\left(-\nu _1^2+\nu_1\nu _2-\nu _2^2\right),\qquad\beta _1= -\frac{1}{54}\left(\nu _1-2 \nu _2\right)\left(2 \nu _1-\nu _2\right)\left(\nu _1+\nu _2\right).
\end{aligned}
\end{equation}
\end{subequations}

Now we intend to show the application of the above result.
\begin{example}\label{largestGUE}
Regarding the deformed Hermite weight with one jump,
$$w(x;t)=e^{-x^2}\left(1-\frac{\beta}{2}+\beta\theta(x-t)\right)=\begin{cases}
   \left(1+\frac{\beta}{2}\right)e^{-x^2},       & \quad \text{if } x>t\\
   \left(1-\frac{\beta}{2}\right)e^{-x^2},  & \quad \text{if }  x\leq t\\
\end{cases},$$
the following formulas are established in \cite{ChenPruessner2005} (The $\tilde{x}$ there has been repalced here with $t$),
\begin{equation}\label{formulasknown}
\begin{aligned}
r_n^2(t)&=2(n+r_n(t))\alpha_n\alpha_{n-1},\\
r_n'(t)&=2(n+r_n(t))(\alpha_{n-1}-\alpha_n),\\
\frac{d}{dt}\ln D_n(t)&=2tr_n(t)-2(n+r_n(t))(\alpha_n+\alpha_{n-1}),\\
\frac{d^2}{dt^2}\ln D_n(t)&=2r_n(t).
\end{aligned}
\end{equation}
Recall here that
\bean
r_n(t)&=&\beta\frac{P_{n}(t,t)P_{n-1}(t,t)}{h_{n-1}}e^{-t^2},\\
D_n(t)&=&\det\left(\int_{-\infty}^\infty x^{i+j}w(x;t)dx\right)_{i,j=0}^{n-1},
\eean
with
\bean
h_n\delta_{mn}&=&\int_{-\infty}^{\infty}P_m(x)P_{n}(x)w(x;t)dx,\\
zP_n(z)&=&P_{n+1}(z)+\alpha_nP_n(z)+\beta_nP_{n-1}(z).
\eean
From (\ref{formulasknown}), we can derive the $\sigma$-form of $P_{IV}$
for
$$ \Xi(t):=\frac{d}{dt}\ln D_n(t)$$
with
\bea\label{exlargestGUEparameter}
\nu_1=0,\qquad\nu_2=2n.
\eea
Hence the the ODE for
$$\varrho(t)=r_n(t)$$
follows immediately from (\ref{generalChazy1})
and, as a consequence, the Chazy's equation (\ref{generalChazyII1}) holds for
$$v(z)=-\frac{1}{2}\left.r_n(t)\right|_{t=\frac{z}{\sqrt{2}}}-\frac{n}{3}$$
with
\bea\label{exlargestGUEChazypara}
\alpha_1= -\frac{2}{3}n^2,\qquad\beta _1=-\frac{8}{27}n^3.
\eea
\begin{remark}
Denote by $\mathbb{P}_{\max}(n,t)$
the largest eigenvalue distribution of GUE on $(0,t)$ and by $\mathbb{P}_{\min}(n,t)$ the smallest one on $(t,\infty).$
Then we readily see in this last example that
$$\Xi(t)=\frac{d}{dt}\ln\mathbb{P}_{\max}(n,t)\qquad\text{or}\qquad-\frac{d}{dt}\ln\mathbb{P}_{\min}(n,t)$$
according as
$$\beta=-2\qquad \text{or}\qquad2.$$
Indeed, for instance, for $\beta=-2,$ we have
\bean
\bar{D}_n(t)&:=&\det\left(\int_{-\infty}^t x^{i+j}e^{-x^2}dx\right)_{i,j=0}^{n-1}=\frac{D_n(t)}{2^n},\\
\frac{\bar{D}_n(t)}{\bar{D}_n(\infty)}&=&\mathbb{P}_{\max}(n,t),
\eean
so that
$$
\frac{d}{dt}\ln\mathbb{P}_{\max}(n,t)=\frac{d}{dt}\ln \bar{D}_n(t)=\frac{d}{dt}\ln D_n(t)=\Xi(t).
$$
In addition, we notice that
$$
r_n(t)=-\frac{P_{n}(t,t)P_{n-1}(t,t)}{\hbar_{n-1}}e^{-t^2},$$
with
$$
\int_{-\infty}^{t}P_m(x)P_{n}(x)e^{-x^2}dx=\frac{1}{2}h_n\delta_{mn}=:\hbar_n\delta_{mn}.$$
\end{remark}
\end{example}

\begin{example}
The $\sigma$-form of $P_{IV}$ with parameters given by (\ref{exlargestGUEparameter}) is also satified by
$$\Xi(t):=\frac{d}{dt}\ln\mathbb{P}(n,t),$$
where $\mathbb{P}(n,t)$ is the probability that $(-t,t)$ is free of eigenvalues in GUE.
See \cite{CaoChenGriffin2014} and \cite{Tracy1994+}.
Since
$$r_n(t):=2\frac{P_{n}(t,t)P_{n-1}(t,t)}{h_{n-1}}e^{-t^2}=\frac{\Xi'(t)}{2}=\frac{\varrho(t)}{2},$$
equation (\ref{generalChazyII1}) is valid for
$$v(z)=-\left.\varrho(t)\right|_{t=\frac{z}{\sqrt{2}}}-\frac{n}{3}$$
with parameters given by (\ref{exlargestGUEChazypara}).
This result is in agreement with the one in \cite{CaoChenGriffin2014} found by using the Riccati equations for $r_n$ and $R_n.$
\end{example}
\end{subsection}

\begin{acknowledgement} 
The financial support of the Macau Science and Technology Development Fund under grant number FDCT 077/2012/A3, FDCT 130/2014/A3
is gratefully acknowledged. We also like to thank the University of Macau for generous support: MYRG 2014--00011 FST, MYRG 2014--00004 FST. 
\end{acknowledgement}

\end{document}